\def\year{2017}\relax
\documentclass[letterpaper]{article}
\usepackage{aaai17}
\usepackage{times}
\usepackage{helvet}
\usepackage{courier}
\usepackage{setspace}
\usepackage[british]{babel}
\usepackage{amsmath}
\usepackage{amssymb,bbm}
\usepackage{amsthm}
\usepackage{graphicx}%
\usepackage{epstopdf}
\usepackage{color}
\usepackage{algpseudocode}
\usepackage{algorithm}

\def\BState{\State\hskip-\ALG@thistlm}
\makeatletter  
 \renewcommand{\ALG@name}{Mechanism} 
\makeatother   
\newtheorem{theorem}{Theorem}
\newtheorem{Lemma}[theorem]{Lemma}
\newtheorem{assumption}[theorem]{Assumption}

\newtheorem{definition}[theorem]{Definition}

\ifodd 0
\newcommand{\rev}[1]{#1}
\newcommand{\frev}[1]{{\color{magenta}#1}}
\newcommand{\com}[1]{\textbf{\color{red}(COMMENT: #1)}} 
\newcommand{\clar}[1]{\textbf{\color{green}(NEED CLARIFICATION: #1)}}
\newcommand{\response}[1]{\textbf{\color{magenta}(RESPONSE: #1)}} 
\else
\newcommand{\rev}[1]{#1}
\newcommand{\frev}[1]{#1}
\newcommand{\com}[1]{}
\newcommand{\clar}[1]{}
\newcommand{\response}[1]{}
\fi
\newcommand{\RNum}[1]{\uppercase\expandafter{\romannumeral #1\relax}}

\newcommand{\squishlist}{
   \begin{list}{$\bullet$}
    { \setlength{\itemsep}{0pt}      \setlength{\parsep}{3pt}
      \setlength{\topsep}{3pt}       \setlength{\partopsep}{0pt}
      \setlength{\leftmargin}{1.5em} \setlength{\labelwidth}{1em}
      \setlength{\labelsep}{0.5em} } }
\newcommand{\squishend}{  \end{list}  }

\usepackage{times}
\begin{document}
\title{Sequential Peer Prediction: Learning to Elicit Effort using Posted Prices}

\author{Yang Liu and Yiling Chen \\ 
Harvard University, Cambridge MA, USA\\
\{yangl,yiling\}@seas.harvard.edu}
\maketitle
\begin{abstract}
Peer prediction mechanisms are often adopted to elicit truthful contributions from crowd workers when no ground-truth verification is available. Recently, mechanisms of this type have been developed to incentivize effort exertion, in addition to truthful elicitation. In this paper, we study a sequential peer prediction problem where a data requester wants to dynamically determine the reward level to optimize the trade-off between the quality of information elicited from workers and the total expected payment. In this problem, workers have homogeneous expertise and heterogeneous cost for exerting effort, both unknown to the requester. We propose a sequential posted-price mechanism to dynamically learn the optimal reward level from workers' contributions and to incentivize effort exertion and truthful reporting. We show that (1) in our mechanism, workers exerting effort according to a non-degenerate threshold policy and then reporting truthfully is an equilibrium that returns highest utility for every worker, and (2) The regret of our learning mechanism w.r.t. offering the optimal reward (price) is upper bounded by $\tilde{O}(T^{3/4})$ where $T$ is the learning horizon. We further show the power of our learning approach when the reports of workers do not necessarily follow the game-theoretic equilibrium. 

\end{abstract}

\section{Introduction}\label{sec:intro}
Crowdsourcing has arisen as a promising option to facilitate machine learning via eliciting useful information from human workers. For example, such a notion has been widely used for labeling training samples, e.g., Amazon Mechanical Turk. Despite its simplicity and popularity, one salient feature or challenge of crowdsourcing is the lack of evaluation for the collected answers, because ground-truth labels often are either unavailable or too costly to obtain. This problem is called {\em information elicitation without verification} (IEWV)~\cite{Waggoner:14}. A class of mechanisms, collectively called {\em peer prediction}, has been developed for the IEWV problem~\cite{Prelec:2004,MRZ:2005,jurca2007collusion,jurca2009mechanisms,witkowski2012robust,witkowski2012peer,radanovic2013}.  In peer prediction, an agent is rewarded according to how his answer compares with those of his peers and the reward rules are designed so that everyone truthfully reporting their information is a game-theoretic equilibrium. 


More recent work \cite{Witkowski_hcomp13,dasgupta2013crowdsourced,2016arXiv160303151S} on peer prediction concerns \emph{effort elicitation}, where the goal is not only to induce truthful report, but also to induce high quality answers by incentivizing agents to exert effort. In such work, the mechanism designer is assumed to know workers' expertise level and their cost for effort exertion and designs reward rules to induce optimal effort levels and truthful reporting at an equilibrium. 

This paper also focuses on the effort elicitation of peer prediction. But different from prior work, our mechanism designer knows neither workers' expertise level nor their cost for effort exertion. We introduce a \emph{sequential peer prediction} problem, where the mechanism proceeds in rounds and the mechanism designer wants to learn to set the optimal reward level (that balances the amount of effort elicited and the total payment) while observing the elicited answers in previous rounds. There are several challenges to this problem. First, effort exertion is not observable and no ground-truth answers are available for evaluating contributions. Hence, it is not immediately clear what information the mechanism designer can learn from the observed answers in a sequential mechanism. Second, forward-looking workers may have incentives to mislead the learning process, hoping for better future returns.

The main contributions of this paper are the following: (1) We propose a \emph{sequential peer prediction} mechanism by combining ideas from peer prediction with multi-armed bandit learning \cite{LR85,Auer:2002:FAM:599614.599677}. (2) In this mechanism, workers exerting effort according to a non-degenerate threshold policy and then reporting truthfully in each round is an equilibrium that returns highest utility for every worker. 
(3) We show that the regret of this mechanism w.r.t. offering the optimal reward is upper bounded by $\tilde{O}(T^{3/4})$ where $T$ is the learning horizon. We also show that under a ``mean-field'' assumption, the sequential learning mechanism can be extended to a setting where workers may not be fully rational. (4) Our sequential peer prediction mechanism is {\em minimal} in that reported labels are the only information we need from workers. 

In the rest of the paper, we first survey the most related work in Section \ref{sec:related}. Section \ref{sec:formulate} introduces our problem formulation. We then present a game-theoretic analysis of worker behavior in a one-stage static setting in Section \ref{sec:static}. Based on the equilibrium analysis of the one-stage setting, we propose and analyze a learning mechanism to learn the optimal bonus level using posted price in Section \ref{sec:learn}. We also discuss an extension of our learning mechanism to a setting where workers may not be fully rational. Section \ref{sec:conclude} concludes this paper. All omitted details can be found in the full version of the paper \cite{fullversion_aaai17}. 

\subsection{Related work}\label{sec:related}


Eliciting high-quality data from effort-sensitive workers hasn't been addressed within the literature of peer prediction until recently. 
\citeauthor{Witkowski_hcomp13} [\citeyear{Witkowski_hcomp13}]  and \citeauthor{dasgupta2013crowdsourced} [\citeyear{dasgupta2013crowdsourced}] formally introduced costly effort into models of IEWV. The costs for effort exertion were assumed to be homogeneous and known and static, one-shot mechanisms were developed for effort elicitation and truthful reporting. Our setting allows participants to have heterogeneous cost of effort exertion drawn from a common unknown distribution and hence we consider a sequential setting that enables learning over time.  \citeauthor{fullversion} [\citeyear{fullversion}]  is the closest to this work. It considered the same general setting and partially resolved the problem of learning the optimal reward level sequentially. There are two major differences however. First, the method developed in \citeauthor{fullversion} [\citeyear{fullversion}] required workers to report their private cost in addition to their answer, which is arguably undesirable for practical applications. 
Our learning mechanism in contrast is "minimal" \cite{segal2007communication,Witkowski_ec13} and only asks for answers (for tasks) from workers. 
Second, the mechanism of \citeauthor{fullversion} [\citeyear{fullversion}] was built upon the output agreement mechanism as the single-round mechanism. Output agreement and hence the mechanism of \citeauthor{fullversion} [\citeyear{fullversion}] suffer from potential, simple collusions of workers: colluding by reporting an uninformative signal will lead to a better equilibrium (higher utility) for workers. By building upon the mechanism of \citeauthor{dasgupta2013crowdsourced} [\citeyear{dasgupta2013crowdsourced}], which as a one-shot mechanism is resistant to such simple collusion, we develop a collusion-resistant sequential learning mechanism. 

\com{Minor: frongillo is not a good citation for minimal mechanisms. May want to consider citing a paper about minimal mechanisms in general, not just minimal peer prediction mechanism.}


Generally speaking, when there is a lack of knowledge of agents, the design problem needs to incorporate learning from prior outputs from running the mechanism -- see \citeauthor{chawla2014mechanism} [\citeyear{chawla2014mechanism}] for specific examples on learning with auction data. And this particular topic has also been studied within the domain of crowdsourcing. For example, \frev{\citeauthor{roth2012conducting} [\citeyear{roth2012conducting}] and \citeauthor{abernethy2015actively} [\citeyear{abernethy2015actively}] consider strategic data acquisition for estimating the mean and for online learning respectively. } \frev{ Our problem differs from above in that both agents' action (effort exertion) and ground-truth outcomes are unavailable.}


\section{Problem Formulation}\label{sec:formulate}


\subsection{Formulation and settings}

Suppose in our system we have one data requester (or a mechanism designer), and there are $N$ candidate workers denoted by $\mathcal C = \{1,2,...,N\}$, where $N \geq 4$. In all we have $N+1$ interactive agents. The data requester has binary-answer tasks, with answer space $\{-1,+1\}$, that she'd like to get labels for. The requester assigns tasks to workers.

Label generated by worker $i \in \mathcal C$ comes from a distribution that depends both on the ground-truth label and an effort variable $e_i$. Suppose there are two effort levels, High and Low, that a worker can potentially choose from: $e_i \in \{H,L\}$. We model the cost $c$ for exerting High effort for each (worker,~task) pair as drawn from a distribution with c.d.f. $F(c)$ on a bounded support $[0, c_{\max}]$; while exerting Low effort incurs no cost. We assume such costs are drawn in an i.i.d. fashion. 
Denote worker $i$'s probability of observing $s' \in \{-1,+1\}$ when the ground-truth label is $s \in \{-1,+1\}$ as $p_{i,e_i} = \Pr(s'=s|s,e_i)$, under effort level $e_i$. 
Note with above we have assumed that the labeling accuracy is symmetric, and is independent of the ground-truth label $s$. Further for simplicity of analysis, we will assume all workers have the same set of $p_{i,H}, p_{i,L}$, denoting as $p_H, p_L$. With higher effort, the expertise level is higher: $1 \geq p_H >p_L \geq 0.5$ -- we also assume the labeling accuracy is no less than 0.5. 
\frev{The above are common knowledge among workers, while the mechanism designer doesn't know the form of $F(\cdot)$; neither does she know $p_H,p_L$. But we assume the mechanism designer knows the structural information, such as costs are i.i.d., workers are effort sensitive, and there are two effort levels etc.} 
\com{It's a little unclear what the mechanism designer knows here. I think the mechanism designer needs to know the structure, but not the exact distribution and probabilities. For example, the mechanism designer doesn't know the form of $F$, but needs to know that the cost is drawn i.i.d. from it. Is it right? May want to be a little more precise here.}

\rev{The goal of the learner is to design a sequential peer prediction mechanism for effort elicitation via observing contributed data from workers, such that the mechanism will help the learner converge to making the optimal action (will be defined later). }


\subsection{Reward mechanism}
Once assigned a task, worker $i$ has a guaranteed base payment $b>0$ for each task he completes. In addition to the base payment, the worker receives a reward $B_i (k)$ for task $k$ that he has provided an answer for. The reward is determined using the mechanism proposed by \citeauthor{dasgupta2013crowdsourced} [\citeyear{dasgupta2013crowdsourced}], where in this paper, we denote this specific peer prediction mechanism for effort elicitation as (\texttt{DG13}). In this mechanism, for each (worker, task) pair $(i,k)$, first a reference worker $j \neq i$ is selected randomly from the set of workers who are also assigned task $k$. \rev{Suppose any pair of such workers have been assigned $d$ other distinct tasks $\{i_1,...,i_d\}$ and $\{j_1,...,j_d\}$ respectively. Then the mechanism pays $B'_i(k)$ to worker $i$  on task $k$ in the following way: the payment consists of two terms; one term that rewards agreement on task $k$, and another that penalizes on uninformative agreement on other tasks: 
\begin{align}
B'_i(k) = 1&\biggl (L_i(k)=L_j(k)\biggr) - L^d_i \cdot L^d_j -\overline{L}^d_i \cdot \overline{L}^d_i,~\label{bonus:index}
\end{align}
where we denote reports from worker $i$ on task $n$ as $L_i(n)$ and $L^d_i := \sum_{n=1}^{d}(1+ L_i(i_n))/(2d), \overline{L}^d_i=1-L^d_i$. Our bonus rule follows exactly the same idea except that we will multiply $B'_i(k)$ by a constant $B \in [0,\bar{B}]$ (which we can choose): $B_i(k) := B\cdot B'_i(k)$.
}
\com{I changed the reward in DG13 to be $B'_i(k)$ to avoid overloading the notation $B_i(k)$.}

\paragraph{Task assignment:} We'd like all workers to work on the same number of tasks, all tasks are assigned to the same number of workers and any pair of workers have distinct tasks. In particular, each worker is assigned $M>1$ tasks -- denote the set of tasks assigned to worker $i$ as $\mathcal T_i: |\mathcal T_i|=M$. This is to simplify the computation of workers' utility and payment functions. Each task is assigned to at least two workers. For any pair of workers who has been assigned a common task, they have at least $1 \leq d < M$ distinct tasks. 
These are to ensure that the (\texttt{DG13}) mechanism is applicable. 
We also set the number of assignments for each task to be the same -- denote this number as $1 \leq K < N$, so that when we evaluate the accuracy of aggregated labels later, all tasks receive the same level of accuracy. But note we do not assign all tasks to all workers, i.e., $K \neq N$. 
Suppose each assigned task $k$ appears in $D \leq d$ tasks' distinct set for each worker. The described assignment can be achieved by assigning $N$ different tasks in each round. For more details on assignments please refer to our full version.

%


\subsection{Worker model}


After receiving each task $k$, worker $i$ first realizes the cost $c_i(k)$ for exerting High effort. Then worker $i$ decides his effort level $e_i(k) \in \{H,L\}$ and observes a signal $L_i(k)$ (label of the task). Worker $i$ can decide either to truthfully report his observation $r_i(k) = 1$ (denote by $r_i(k)$ the decision variable on reporting) or to revert the answer $r_i(k) = 0$:
 \[ L^r_i(k) =\left\{
                \begin{array}{ll}
                 L_i(k), ~\text{if}~r_i(k)=1\\
                 -L_i(k), ~\text{if}~r_i(k)=0
                \end{array}
              \right.
 \]

Workers are utility maximizers.
Denote the utility function at each time (or step) for each worker as $u_i$, which is assumed to have the following form (payment $-$ cost):
\begin{align*}
u_i = Mb + \sum_{k \in \mathcal T_i} B_i(k) - \sum_{k \in \mathcal T_i} c_i(k),~\forall i.
\end{align*}

\subsection{Data requester model}

After collecting labels for each task, the data requester will aggregate labels via majority voting. Denote workers who labeled task $k$ as $w_k(1),...,w_k(K)$. Then the aggregate label for $k$ is given by 
$$
L^A(k) = 1\biggl(\sum_{n=1}^K L^r_{w_k(n)}(k)/K > 0\biggr) \cdot 2-1.~
$$

\com{I think the indicator function gives value 0 or 1, but our labels are -1 and 1. So, when the indicator function returns 0, we need $L^A(k)$ to be -1.}
 
The data requester's objective is to find a bonus level $B$ (as in $B_i(k)$) that balances the accuracy of labels collected from workers, and the total payment. Denote requester's objective function at each step as $\mathcal U(B)$ (assigning $N$ tasks):
\begin{align*}
\mathcal U(B) := \sum_{k=1}^N \biggl [ \Pr[L^A(k) = L(k)]-\eta \sum_{n=1}^K \mathbb E[B_{w_k(n)}(k)]\biggr],
\end{align*}
 where $L(k)$ denotes the true label of task $k$, and $\eta>0$ is a weighting constant balancing the two terms in the objective.

Since we have assumed that all tasks have been assigned the same number of workers, and workers are homogeneous in their labeling accuracy and cost (i.i.d.), we know all tasks enjoy the same probability of having a correct label (a-priori). We denote this probability as $\mathcal P(B) := \Pr[L^A(k) = L(k)], \forall k$. Further as workers do not receive payment when a task is not assigned to him, $\mathcal U(B)$ can be simplified (and normalized \footnote{which does not affect optimizing the utility function.}) to the following form:
\begin{align}
\mathcal U(B) = \mathcal P(B) - \frac{\eta}{N} \sum_{i \in \mathcal C} \sum_{k=1}^N \mathbb E[B_i(k)],~\label{learner:objective}
\end{align}
Suppose there exists a maximizer
$
B^* = \text{argmax}_B \mathcal U(B)~.
$

\subsection{Sequential learning setting}
\rev{Suppose our sequential learning algorithm goes for $T$ stages. At each stage $t=1,...,T$, learner assigns a certain number of tasks $M_i(t)$ to a set of selected workers $i \in S(t)$\footnote{For details please refer to our algorithm.}. The learner offers a bonus bundle $B_{i,t}$ to each worker $i \in S(t)$ (the bonus constant in reward mechanism). The regret of offering $\{B_{i,t}\}_{i,t}$ w.r.t. $B^*$ is defined as follows:
\begin{align}
R(T) = T\cdot\mathcal U(B^*)- \sum_{t=1}^T  \sum_{i \in S(t)}\frac{M_i(t) \cdot \mathbb E[\mathcal U(B_{i,t})] }{\sum_{j \in S(t)} M_j(t)}.\label{regret:defn}
\end{align}
}
Note we normalize $\mathbb E[\mathcal U(B_{i,t})]$ using the number of assignments -- intuitively the more the requester assigned with a wrong price, the more regret will be incurred. The goal of the data requester is to design an algorithm such that $R(T) = o(T)$. We can also define $R(T)$ as being un-normalized, which will add a constant (bounded number of assignments at each step) in front of our results later. 


\section{One stage game-theoretic analysis}\label{sec:static}


From the data requester's perspective, we need to first understand workers' actions towards effort exertion and reporting under different bonus levels, in order to figure out the optimal $B^*$. We start with the case that the data requester knows the cost distribution, and we characterize the equilibria for effort exertion and reporting, i.e. $(\mathbf{e_i,r_i})$, on workers' side. Note $\mathbf{e_i,r_i}$ are both vectors defined over all tasks -- this is a simplification of notation as workers do not receive all tasks. We are safe as if task $k$ is not assigned to $i$,  worker $i$ does not make decisions on $(e_i(k),r_i(k))$. We define Bayesian Nash Equilibrium (BNE) in our context as follows:
\begin{definition}
We say $\{(\mathbf{e^*_i,r^*_i})\}_{i \in \mathcal C}$ is a BNE if $\forall j, (\mathbf{e_j,r_j})$:
\begin{align*}
\mathbb E[u_j| \{(\mathbf{e^*_i,r^*_i})\}_{i \in \mathcal C}] \geq 
\mathbb E[u_j| \{(\mathbf{e^*_i,r^*_i})\}_{i \neq j},(\mathbf{e_j, r_j})]~. 
\end{align*}
\end{definition}
\frev{In this paper, we restrict our attention to symmetric BNE.}
For each assigned task, we have a Bayesian game among workers in $\mathcal C$: a worker's decision on effort exertion is a function of $\mathbf{c_i}$, $\mathbf{e_i(c_i)}: [0, c_{\text{max}}]^M \rightarrow \{H, L\}^M$, which specifies the effort levels for worker $i$ when his realized cost is $\mathbf{c_i}$ and $\mathbf{r_i(e_i)}: \{0, 1\}^M \rightarrow \{0, 1\}^M$ gives the reporting strategy for the chosen effort level. We focus on threshold policies: that is, there is a threshold $c^*$ such that $e_i(k) =H$ for all $c_i(k) \leq c^*$ and $e_i(k) =L$ otherwise. \frev{In fact, players must play a threshold strategy for effort exertion at any symmetric BNE: workers' outputs do not depend on $c_i(k)$ and worker $i$'s chance of getting a bonus will not change when he has a cost $c'_i(k) < c_i(k)$; so a worker will choose to exert effort, if it is a better move for an even higher cost.} We will use $(c^*,r_i(k))$ to denote this threshold ($c^*$) strategy for workers. Denote $r_i(\cdot) \equiv 1$  the reporting strategy that $r_i(H) = r_i(L) =1$, i.e. reporting truthfully regardless of the choice of effort.


\begin{theorem}
When $p_L>0.5$ and $F(c)$ is concave, there exists a unique threshold $c^*(B)>0$ for $B>0$ such that $(c^*(B), 1)$ is a symmetric BNE for all workers on all tasks.\label{bne:static}
\end{theorem}


\paragraph{Other equilibrias:} The above threshold policy is unique only in non-degenerate effort exertion ($c^*>0$). There exist other equilibria. We summarize them here:
\begin{itemize}
\item \emph{Un-informative equilibrium:} Colluding by always reporting the same answer to all tasks is an equilibrium. Similarly as mentioned in \cite{dasgupta2013crowdsourced}, when colluding (pure or mixed strategies) the bonus index defined in Eqn. (\ref{bonus:index}) reduces to 0 for each worker, which leads to a worse equilibrium. 
\item \emph{Low effort:} When $p_L = 1/2$, $c^*=0$, i.e., no effor exertion (followed by either truthful or untruthful reporting) is also an equilibrium: when no one else is exerting effort, each worker's answer will be compared to a random guess. So there would be no incentive for effort exertion.
\item \emph{Permutation:} Exerting the same amount of effort and then reverting the reports ($r_i \equiv 0$) is also an equilibrium. 
\end{itemize}
But we would like to note that though there may exist multiple equilibria, all others lead to strictly less utility for each worker at equilibrium compared to the threshold equilibrium with $c^*>0$ followed by truthful reporting, except for the permutation equilibria, which gives the same expected utility to workers. \com{May want to mention that the permutation equilibrium gives the same utility/expected utility.}

\emph{Solve for optimal $B^*$:} After characterizing the equilibria $c^*$ on effort exertion as a function of $B$, we can compute $\mathcal P(B)$ and $ \mathbb E[B_i(k)]$ for each reward level $B$. Then solving for the optimal reward level becomes a programming problem in $B$, which can be solved efficiently when certain properties, e.g. convexity, can be established for $\mathcal U(\cdot)$. 

\section{Sequential Peer Prediction}\label{sec:learn}


In this section we propose an adaptive learning mechanism to learn to converge to the optimal or nearly optimal reward level. As mentioned earlier, a recent work \cite{fullversion} attempted to resolve this challenge. But besides the output labels, workers are also required to report the private costs, in which sense the proposed learning mechanism is not ``minimal''. We try to remove this requirement by learning only through the label information reported by the workers. In this section, we assume the requirements for Theorem \ref{bne:static} hold, and workers will follow an equilibrium that returns the highest utility.

\subsection{Challenges}

In designing the mechanism, we face two main challenges. The first challenge is on the learning part. In order to select the best $B^*$, we need to compute $\mathcal U(B), \forall B$, which can be computed as a function of $B$ and $\bar{p}(B)$, the probability of labeling accurately when $B$ is offered and the threshold policy $c^*(B)$ is adopted by workers:
\begin{align}
\bar{p}(B):=F(c^*(B))p_H+[1-F(c^*(B))]p_L. \label{pb}
\end{align}
The dependency on $B$ is straight-forward. For $\bar{p}(B)$, e.g. when using Chernoff bound for approximating $\mathcal P(B)$:
\begin{align*}
\mathcal P(B)&= \Pr\biggl[\frac{\sum_i 1(\text{worker i is correct})}{M} \geq 0.5\biggr]\\
&\geq 1-\text{exp}(-2(\bar{p}(B)-0.5)^2 M),
\end{align*}
it is clear $\mathcal P(B)$ is a function of $\bar{p}(B)$. \frev{In fact both $\mathbb E[B_i(k)]$ and $\mathcal P(B)$ are functions of $\bar{p}(B)$, so is $\mathcal U(\cdot)$. For details, please see Appendix of \cite{fullversion_aaai17}. } The question pings down to learn $\bar{p}(B)$. Since we do not have the ground-truth labels, we have no way to directly evaluate $\bar{p}(B)$ via checking workers' answers. Also since we do not elicit reports on private costs, we are un-able to estimate the amount of induced effort for each reward level.  

The second challenge we have is that when workers are returning and participating in a/an sequential/adaptive learning mechanism, they have incentives to mislead the learning process by deviating from the one-shot BNE strategy for a task, so to create untruthful samples (and then collected by learner), which will lead the learner into believing that inducing certain amount of effort requires a much higher reward level. The cost-reporting mechanism described in \cite{fullversion}  proposes a method to deter such a deviation by eliminating workers who over-reported from receiving potentially higher bonus. We will describe a two-fold cross validation approach to decouple such incentives, which aims to remove the requirement of reporting additional information. 

\subsubsection{Learning w/o ground-truth}
The following observation inspires our method for learning without ground-truth. For each bonus level $B$, we can estimate $\bar{p}(B)$ (at equilibrium) through the following experiments: the probability of observing a pair of matching answers for any pair of workers $i,j$ (denoted by $p_{m}(B)$ for each bonus level $B$) on equilibrium can be written as follows:
\begin{align}
p_m(B) &= \underbrace{\bar{p}^2(B)}_{\text{match on correct label}}+\underbrace{(1-\bar{p}(B))^2}_{\text{match on wrong label}}. \label{infer:equilibria}
\end{align}
The above matching formula forms a quadratic equation of $\bar{p}(B)$. 
From Eqn. (\ref{pb}) we know
$
\bar{p}(B) \geq 0.5,~\forall B,~\text{when~} p_H, p_L \geq 0.5
$. Then the only solution to the matching Eqn. (\ref{infer:equilibria}) that is larger than $0.5$ is
\begin{align*}
\bar{p}(B) =1/2+\sqrt{2p_m(B)-1}/2.
\end{align*}
Above solution is well defined, as from  Eqn. (\ref{infer:equilibria}) we can also deduce that $p_m(B) \geq 1/2$. Therefore though we cannot evaluate each worker's labeling accuracy directly, we can make such an inference using the matching probability, which is indeed observable.

\subsubsection{Decoupling incentives via cross validation}

To solve the incentive issue, we propose the following cross validation approach (illustrated in Fig. \ref{mechanism:ill}). First the entire crowd $\mathcal C$ is separated into two groups $G_1, G_{-1}$ uniformly random, but with equal size (when $N$ is even) or their sizes differ by at most 1 (when $N$ is odd). Suppose we have at least $N \geq 4$. Denote worker $i$'s group ID as $g(i) \in \{-1,1\}$. Then we have $|G_1|, |G_{-1}| \geq 2$. For our learning algorithm, only the data/samples collected from group $-g(i)$ will be used to reward any worker $i$ in group $g(i)$. Secondly when selecting reference worker for comparing answers for mechanism (\texttt{DG13}), we select from the same group $g(i)$. 
\begin{figure}[!ht]
\centering
\includegraphics[width=0.4\textwidth]{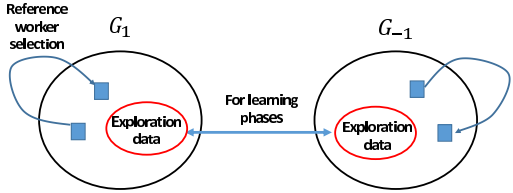}
\caption{Illustration of our mechanism.}\label{mechanism:ill} 
\end{figure}


\subsection{Mechanism}
We would like to treat each bonus level as an ``arm'' (as in standard MAB context) to explore with. 
Since we have a continuous space of bonus level $B$, we separate the support of bonus level $B$ ($[0,\bar{B}]$) into finite intervals. Then we treat each bonus interval as an arm. Our goal is to select the best one of them, and bound the performance in such a selection. 

\com{I don't see why we need the z superscript below. I could have missed something.}We set up $\lceil T^z \rceil$ arms as follows:  chooses a $0<z<1$, separate $[0,\bar{B}]$ into $N_a = \lceil T^z \rceil$ uniform intervals:
$$
[0, \bar{B}/N_a], ..., [(k-1)\bar{B}/N_a, k\bar{B}/N_a],....,[(N_a-1)\bar{B}/N_a, \bar{B}]
$$ For each interval we take its right end point as the bonus level to offer:
$B_k =  k\bar{B}/N_a.$ Denote by $\tilde{p}^{g}_{m,t}(B_k)$ the estimated matching probability for agents in group $g$ under bonus level $B_k$, and $\tilde{p}^g_t(B_k)$s the estimated $\bar{p}(B_k)$ for group $g$, at stage $t$; and we use $\tilde{\mathcal U}_{\tilde{p}}(B)$ to denote the estimated utility function when using a noisy $\bar{p}(B)$ ($\tilde{p}$), instead of the true ones.
We present Mechanism \ref{m:simple}.\footnote{We assume we know $p_L$ or a non-trivial lower bound on $p_L>0.5$.}
\begin{algorithm}[!ht]
\caption{ (\texttt{SPP\_PostPrice})}\label{m:simple}
\begin{algorithmic}
\State \textbf{Initialization:} $t=1$. $D(t):=t^{\theta}\log t, 0\leq \theta\leq 1$. Explore each bonus level $B_k$ once and update $\tilde{p}^{g}_{m,t}(B_k)$ (for details please refer to the exploration phases), and set the number of explorations as $n_i(t) = 1$.
\For{$t=1$ to $T$}
\State Set $\mathcal E(t):=\{i: n_i(t) < D(t)\} $. 
\If{ \textbf{$\mathcal E(t) \neq  \emptyset$}}
\State \emph{Exploration}: randomly pick $B_k, k \in \mathcal E(t)$ to offer.
\State Follow subroutine (\texttt{Explore\_Crowd}).
\Else{}
\State \emph{Exploitation}: $S(t) = \mathcal C$.  Offer $B^{*}_{g,t}$ to $i \in G_g$:
\begin{align*}
B^{*}_{g,t} = \text{argmax}_{B_k, k=1,2,...,\lceil T^z \rceil} \tilde{\mathcal U}_{\tilde{p}^g_t(B_k)}(B_k),
\end{align*}
\State Follow (\texttt{DG13}) for workers in each group $G_g$.
\EndIf
\EndFor
\end{algorithmic}
\end{algorithm}

\begin{algorithm}[!ht]
\caption{ (\texttt{Explore\_Crowd})}\label{m:sub}
\begin{algorithmic}
\State At exploration phases,
\begin{itemize}
\item[1:] Randomly select two workers $S_g(t) = \{i^g(t),j^g(t)\}$ from each group $G_g$; $S(t) := \cup_{g} S_g(t)$.
\item[2:] Assign 1 common and $d$ distinct tasks to each $S_g(t)$.
\item[3:] Denote by $\mathcal E_k(t)$ the set of time steps the algorithm enters exploration and offered price $B_k$ by time $t$. Estimate the probability of matching $\tilde{p}^{g}_{m,t}(B_k)$ for each crowd $G_g, g \in \{-1,1\}$ by averaging:
\begin{align*}
\tilde{p}^{g}_{m,t}(B_k) = \overline{\sum_{t' \in \mathcal E_k(t)} 1(L_{i^{-g}(t')}(t') = L_{j^{-g}(t')}(t'))},
\end{align*}
and reset $\tilde{p}^{g}_{m,t}(B_k)$ to $\max\{\tilde{p}^{g}_{m,t}(B_k),1/2\}$. $L_{i}(t')$ denotes the label for the common task at time $t'$.
\item[4:] Compute $\tilde{p}^g_t(B_k)$ (estimate for $\bar{p}(B_k)$ at time $t$):
\begin{align*}
\tilde{p}^g_t(B_k) = 1/2+\sqrt{2\tilde{p}^{g}_{m,t}(B_k)-1}/2.
\end{align*}
Reset $\tilde{p}^g_t(B_k):=\max\{\tilde{p}^g_t(B_k),p_L\}$.
\end{itemize}
\end{algorithmic}
\end{algorithm}

Note since we assign the same number of tasks to each labeler at all stages, we have the regret defined in Eqn. (\ref{regret:defn}) become equivalent with the following form:
\begin{align*}
R(T) = \sum_{t=1}^T |  \sum_{g \in \{1,-1\}} \omega_g(t) \cdot \mathbb E[\mathcal U(B^{*}_{g,t})] - \mathcal U(B^*))|,
\end{align*}
where when $t$ is in exploration stages 
$
\omega_g(t) \equiv 1/2,
$
otherwise
$
\omega_g(t) := |G_g|/N.
$

\subsection{Equilibrium analysis: workers' side}

Denote a worker's action profile at step $t$ as $\mathbf{a_i}(t):=(\mathbf{e_i}(t),\mathbf{r_i}(t))$. We adopt BNE as our solution concept:
\begin{definition}
A set of reporting strategy $\{\mathbf{\tilde{a}_i}:= \{\mathbf{\tilde{a}}_i(t)\}_{t=1,...,T}\}_{i \in \mathcal C}$ is a BNE if for any $i$, $\forall \mathbf{\tilde{a}'_i} \neq \mathbf{\tilde{a}_i}$ we have
\begin{align*}
&\sum_{t} \mathbb E[\max_{\mathbf{e_i}(t),\mathbf{r_i}(t)}u_i(\mathbf{\tilde{a}_i}, \mathbf{\tilde{a}_{-i}})|\mathbf{\tilde{a}_i}(1:t-1), \mathbf{\tilde{a}_{-i}}(1:t-1)] \\
&\geq \sum_{t} \mathbb E[\max_{\mathbf{e_i}(t),\mathbf{r_i}(t)} u_i(\mathbf{\tilde{a}'_i}, \mathbf{\tilde{a}_{-i}})|\mathbf{\tilde{a}'_i}(1:t-1), \mathbf{\tilde{a}_{-i}}(1:t-1)] ~.
\end{align*}
\end{definition}

We first characterize the equilibrium strategy for workers' effort exertion and reporting with (\texttt{SPP\_PostPrice}).
\begin{Lemma}
At a symmetric BNE, the strategy for workers can be decoupled into a composition of myopic equilibrium strategies, that is 
$
e_i(t) = c^*(B^{*,g(i)}(t)),
$
combined with $r_i(t) \equiv 1, \forall i,t$.
\end{Lemma}

\begin{proof}
W.l.o.g., consider worker $i$ in $G_1$. We are going to reason that deviating from one step BNE strategy for effort exertion is non-profitable for worker $i$, when other players are following the equilibria. Since the one stage equilibrium strategy maximizes the utility at current stage, and it does not affect the past utilities that have been already collected, the potential gain by deviating comes from the future gains in utilities in: (1) the offered bonus level (2) matching probability from other peers. For the bonus level offered to worker $i$, it will only be computed using observed data from workers in $G_{-1}$ at exploration phases. Note for our online learning process,  the exploration phases only depend on the pre-defined parameter $D(t)$, which does not depend on worker $i$'s data (deterministic exploration). Similarly for all other workers $j \in G_1$ (reference workers), their future utility gain is not affected by worker $i$'s data. Therefore an unilateral deviation from worker $i$ will not affect the matching probability from other peers. So no deviation  is profitable.
\end{proof}

Again consider colluding workers. Potentially when offered a certain bonus level, workers can collude by not exerting effort regardless of their cost realization, so to mislead the learner into believing that in order to incentivize certain amount of effort, a much higher bonus level is needed.\footnote{This is similar to the colluding strategy that contributes uninformative signals we studied in Section \ref{sec:static}.} There potentially exist infinitely many colluding strategies for workers to game against a sequential learning algorithm. We focus on the following easy-to-coordinate (for workers), yet powerful strategy (\texttt{Collude\_Learn}):
\begin{definition}[\texttt{Collude\_Learn}]
Workers collude by agreeing to exert effort when the offered bonus level $B^{*}_{g(i),t}$ satisfies $B^{*}_{g(i),t} \geq B(c_{i,t}(k)) + \nabla B,~
$where $c_{i,t}(k)$ is the cost for workers $i$ to exert effort for task $k$ at time $t$. $\nabla B > 0$ is a collusion constant. %
\end{definition}
In doing so workers mislead the learner into believing that a higher bonus (differs by $\nabla B$) is needed to induce certain effort. 
The next lemma establishes the collusion-proofness: 
\begin{Lemma}
Colluding in \emph{(\texttt{SPP\_PostPrice})}  via \emph{\texttt{(Collude\_Learn)}} is not an equilibrium. \label{collusion:price}
\end{Lemma}
In fact the reasoning for removing \emph{all} symmetric colluding equilibria is similar -- regardless of how others collude on effort exertion, when a worker's realized cost is small enough, he will deviate. 

\subsection{Performance of (\texttt{SPP\_PostPrice})}
We impose the following assumptions  \footnote{Please refer to our full version for justification. }:
\begin{assumption}
$\mathcal U(\cdot)$ is Lipschitz in both $\bar{p}(B)$ and $B$:
\begin{align*}
|\tilde{\mathcal U}_{\tilde{p}}(\tilde{B})-\mathcal U(B)| \leq L_1  |\tilde{p}(\tilde{B})-\bar{p}(\tilde{B})|+L_2|\tilde{B}-B|,
\end{align*}
 $L_1, L_2>0$ are the Lipschitz constants.
\end{assumption}
%

We have the following theorem:
\begin{theorem}
The regret of \emph{(\texttt{SPP\_PostPrice})} is: 
\begin{align*}
R(T) \leq  & O(\lceil T^{\theta+z}\log T \rceil) + \frac{L_1C(\theta)}{\sqrt{2p_L-1}}T^{1-\theta/2} \\
&+ L_2 C(z)\bar{B} T^{1-z} + \text{const.} 
\end{align*}
$0<\theta,z<1$ are tunable parameters. $C(\theta), C(z) >0$ are constants. The optimal regret is $R(T) \leq \tilde{O}(T^{3/4})$ when setting $z=1/4,\theta=1/2$.
\label{thm:learn}
\end{theorem}



\begin{proof} (Sketch)
First notice by triangular inequality we know
$
R(T) \leq \sum_{g \in \{1,-1\}}\omega_g  R_g(T),
$
i.e., the total regret is upper bounded by the weighted sum of each group's regret. Since the two learning processes for the two groups $G_1,G_{-1}$ parallel each other, we omit the $g$ super- or sub-script. We analyze for the regret incurred at the exploration and exploitation stages respectively. For exploitation regret, we first characterize the estimation error for estimating each $U(B_k)$. First by mean value theorem we can show:
\begin{align*}
|\bar{p}(B_k)-\tilde{p}_t(B_k)| \leq \frac{1}{2\sqrt{2p_L-1}}|p_m(B_k)-\tilde{p}_{m,t}(B_k)|.
\end{align*}
At time $t$, an exploitation phase, there are $D(t)$ number of samples guaranteed; so the estimation $\tilde{p}_{m,t}(B_k)$ satisfies:
$
\Pr[|p_m(B_k)-\tilde{p}_{m,t}(B_k)| \geq \frac{1}{t^{\theta/2}}] \leq  \frac{2}{t^2}~.
$
Then w.h.p., we have established that 
$
|\bar{p}(B_k)-\tilde{p}_t(B_k)|  \leq \frac{ t^{-\theta/2}}{2\sqrt{2p_L-1}}.
$
Further by Lipschitz condition we know
$|\tilde{\mathcal U}(B_k)-\mathcal U(B_k)| \leq \frac{L_1 \cdot t^{-\theta/2}}{2\sqrt{2p_L-1}}.$ For any $B$ that falls in the same interval as $B_k$ we know:
$
|\mathcal U(B)-\mathcal U(B_k)| \leq L_2  |B-B_k|\leq L_2 \bar{B}T^{-z}.
$

Denote by
$B^*_t = \text{argmax}_{B_k, k=1,2,...,\lceil T^z \rceil} \tilde{\mathcal U}_t(B_k)~,
$
i.e., $B^*_t$ is the estimated optimal bonus level at time $t$ -- at any time $t$, by searching through all arm and we can find the one maximizes the utility function from the empirically estimated function the learner currently has. Combine above arguments, we can prove
$
\mathcal U(B^*)-\mathcal U(B^*_t) \leq \frac{L_1 \cdot t^{-\theta/2}}{2\sqrt{2p_L-1}}+L_2 \bar{B}T^{-z}.
$ Then the exploration regret can be bounded as 
\begin{align*}
&\sum_{t=1}^T ( \frac{L_1 \cdot t^{-\theta/2}}{2\sqrt{2p_L-1}}+L_2 \bar{B}T^{-z}) + O(\sum_{t=1}^T \frac{2}{t^2})\\
&\leq \frac{L_1 \cdot C(\theta)}{2\sqrt{2p_L-1}}T^{1-\theta/2} + L_2 \cdot C(z) \bar{B} T^{1-z} + \text{const.} 
\end{align*}
where we have used the fact that for any $0 < \alpha < 1$, there exists a constant $C(\alpha)$ such that
$\sum_{t=1}^T \frac{1}{t^{\alpha}} \leq C(\alpha) T^{1-\alpha}~.
$
The total number of explorations are ($T^{\theta}\log T$ number of explorations needed for each of the $\lceil T^z \rceil$ arms): 
$
\lceil T^z \rceil \cdot T^{\theta}\log T = \lceil T^{\theta+z}\log T \rceil~.
$
Sum over we finish the proof.
%
%

\end{proof}


\subsection{Beyond game theoretical model}\label{sec:nongame}

So far we have modeled the workers as being fully rational, and the reports as coming from game theoretical responses. Consider the case workers' responses do not necessarily follow a game (or arguably no one is fully rational). Instead we assume each worker has a labeling accuracy $p_i(B)$ for different $B$, where $p_i(B)$ can come from different models, being game theory driven, behavioral model driven or decision theory driven, and can be different for different workers. 

\emph{Challenges and a mean filed approach}: With this model, we can again write $\mathcal U(B)$ as a function of $\{p_i(B)\}_i$ and $B$. In order to select $B^*$, again we need to learn $p_i(B)$. We re-adopt the bandit model we described earlier, and estimate $p_i(B)$s via observing the matching probability between worker $i$ and a randomly selected reference worker $j$: For each $i,B$ we define 
$\bar{p}_{-i}(B) := \sum_{j \in G_{g(i)}\backslash i} p_j(B)/|G_{g(i)}\backslash i |, ~\text{and we have } $
\begin{align}
&p^{i}_{m}(B) = p_i(B)\bar{p}_{-i}(B) + (1-p_i(B))(1-\bar{p}_{-i}(B)),\label{eqn:match}
\end{align}
where $p^{i}_{m}(B) $ is the probability of observing a matching for worker $i$, when a random reference worker is drawn (uniformly) from his group. The above forms a system of quadratic equations in $\{p_i(B)\}_i$ when $\{p^{i}_{m}(B)\}_i$s are known. We then need to solve a perturbed quadratic equations for $\{p_i(B)\}_i$, with $\{p^{i}_{m}(B)\}_i$s being estimated via observations (Step 3 of (\texttt{Explore\_Crowd})). The following challenges arise for analysis: (1) it is hard to tell whether the solution for above quadratic equations is unique or not. (2) Solving a set of perturbed (error in estimating $\{p^{i}_{m}(B)\}_i$s) quadratic equations for each $B$ incurs heavy computations.\footnote{We do not claim this is impossible to do. Rather, analyzing the output from such a system of perturbed quadratic equations merits a further study.}

Instead, by observing the availability of relatively large and diverse population of crowd workers, we make the following mean filed assumption:
\begin{assumption}
For any worker $i$,
$\bar{p}_{-i}(B) \equiv \bar{p}_{g(i)}(B)~.$
\end{assumption}
That is one particular worker's expertise level does not affect the crowd's mean.
This is not a entirely unreasonable assumption to make, as the candidate pool of crowd workers is generally large. With above $p^{i}_{m}(B) $ then becomes 
$$p^{i}_{m}(B) = p_i(B)\bar{p}_{g(i)}(B)+ (1-p_i(B))(1-\bar{p}_{g(i)}(B)).$$
Averaging over $i \in G_{g(i)}$ we have:\\
$
\sum_{i \in G_{g(i)}} p^{i}_{m}(B)/|G_{g(i)}| =\bar{p}^2_{g(i)}(B) + (1-\bar{p}_{g(i)}(B))^2,
$ which is very similar to the matching equation we derived earlier on. Again we can solve for $\bar{p}_{g(i)}(B)$ as a function of $\sum_{i \in G_{g(i)}} p^{i}_{m}(B)/|G_{g(i)}|$. 
Plugging $\bar{p}_{g(i)}(B)$ back into Eqn. (\ref{eqn:match}), we obtain an estimate of $p^i(B)$ as follows:
\begin{align*}
p_i(B) = (p^{i}_{m}(B)+\bar{p}_{g(i)}(B)-1)/(2\bar{p}_{g(i)}(B)-1).
\end{align*}
 
Similar regret can be obtained -- the difference only lies in estimating $p_i(B_k)$s. Details can be found in \cite{fullversion_aaai17}.


\section{Conclusion}\label{sec:conclude}

We studied the sequential peer prediction mechanism for eliciting effort using posted price. We improve over status quo towards making the peer prediction mechanism for effort elicitation more practical: (1) we propose a posted-price and ``minimal'' sequential peer prediction mechanism with bounded regret. The mechanism does not require workers to report additional information, except their answers for assigned tasks. Further we show our learning results can generalize to the case when workers may not necessarily be fully rational, under a mean-filed assumption. (2) Workers exerting effort according to an informative threshold strategy and reporting truthfully is an equilibria that returns highest utility.  

\paragraph{Acknowledgement:} We acknowledge the support of NSF grant CCF-1301976.

\bibliographystyle{aaai}
\bibliography{myref,library}

\begin{thebibliography}{}

\bibitem[\protect\citeauthoryear{Abernethy \bgroup et al\mbox.\egroup
  }{2015}]{abernethy2015actively}
Abernethy, J.; Chen, Y.; Ho, C.-J.; and Waggoner, B.
\newblock 2015.
\newblock {Actively Purchasing Data for Learning}.
\newblock In {\em ACM EC 2015}.

\bibitem[\protect\citeauthoryear{Auer, Cesa-Bianchi, and
  Fischer}{2002}]{Auer:2002:FAM:599614.599677}
Auer, P.; Cesa-Bianchi, N.; and Fischer, P.
\newblock 2002.
\newblock {Finite-time Analysis of the Multiarmed Bandit Problem}.
\newblock {\em Machine Learning} 47:235--256.

\bibitem[\protect\citeauthoryear{Chawla, Hartline, and
  Nekipelov}{2014}]{chawla2014mechanism}
Chawla, S.; Hartline, J.; and Nekipelov, D.
\newblock 2014.
\newblock {Mechanism Design for Data Science}.
\newblock In {\em ACM EC 2014},  711--712.

\bibitem[\protect\citeauthoryear{Dasgupta and
  Ghosh}{2013}]{dasgupta2013crowdsourced}
Dasgupta, A., and Ghosh, A.
\newblock 2013.
\newblock {Crowdsourced Judgement Elicitation with Endogenous Proficiency}.
\newblock In {\em WWW 2013},  319--330.

\bibitem[\protect\citeauthoryear{Jurca and Faltings}{2007}]{jurca2007collusion}
Jurca, R., and Faltings, B.
\newblock 2007.
\newblock {Collusion-resistant, Incentive-compatible Feedback Payments}.
\newblock In {\em ACM EC 2007},  200--209.

\bibitem[\protect\citeauthoryear{Jurca and
  Faltings}{2009}]{jurca2009mechanisms}
Jurca, R., and Faltings, B.
\newblock 2009.
\newblock {Mechanisms for Making Crowds Truthful}.
\newblock {\em JAIR} 34(1):209.

\bibitem[\protect\citeauthoryear{Lai and Robbins}{1985}]{LR85}
Lai, T.~L., and Robbins, H.
\newblock 1985.
\newblock {Asymptotically Efficient Adaptive Allocation Rules}.
\newblock {\em Advances in Applied Mathematics} 6:4--22.

\bibitem[\protect\citeauthoryear{Liu and Chen}{2016a}]{fullversion}
Liu, Y., and Chen, Y.
\newblock 2016a.
\newblock {Learning to Incentivize: Eliciting Effort via Output Agreement}.
\newblock In {\em IJCAI 2016}.

\bibitem[\protect\citeauthoryear{Liu and Chen}{2016b}]{fullversion_aaai17}
Liu, Y., and Chen, Y.
\newblock 2016b.
\newblock {Sequential Peer Prediction: Learning to Elicit Effort using Posted
  Prices}.
\newblock {\em AAAI 2017, full version, http://arxiv.org/abs/1611.09219}.

\bibitem[\protect\citeauthoryear{Miller, Resnick, and
  Zeckhauser}{2005}]{MRZ:2005}
Miller, N.; Resnick, P.; and Zeckhauser, R.
\newblock 2005.
\newblock Eliciting informative feedback: {T}he peer-prediction method.
\newblock {\em Management Science} 51(9):1359 --1373.

\bibitem[\protect\citeauthoryear{Prelec}{2004}]{Prelec:2004}
Prelec, D.
\newblock 2004.
\newblock {A Bayesian Truth Serum for Subjective Data}.
\newblock {\em Science} 306(5695):462--466.

\bibitem[\protect\citeauthoryear{Radanovic and Faltings}{2013}]{radanovic2013}
Radanovic, G., and Faltings, B.
\newblock 2013.
\newblock {A Robust Bayesian Truth Serum for Non-Binary Signals}.
\newblock In {\em AAAI 2013}.

\bibitem[\protect\citeauthoryear{Roth and
  Schoenebeck}{2012}]{roth2012conducting}
Roth, A., and Schoenebeck, G.
\newblock 2012.
\newblock {Conducting Truthful Surveys, Cheaply}.
\newblock In {\em ACM EC 2012},  826--843.

\bibitem[\protect\citeauthoryear{Segal}{2007}]{segal2007communication}
Segal, I.
\newblock 2007.
\newblock {The Communication Requirements of Social Choice Rules and Supporting
  Budget Sets}.
\newblock {\em Journal of Economic Theory} 136(1):341--378.

\bibitem[\protect\citeauthoryear{{Shnayder} \bgroup et al\mbox.\egroup
  }{2016}]{2016arXiv160303151S}
{Shnayder}, V.; {Agarwal}, A.; {Frongillo}, R.; and {Parkes}, D.~C.
\newblock 2016.
\newblock {Informed Truthfulness in Multi-Task Peer Prediction}.
\newblock {\em ACM EC 2016}.

\bibitem[\protect\citeauthoryear{Waggoner and Chen}{2014}]{Waggoner:14}
Waggoner, B., and Chen, Y.
\newblock 2014.
\newblock {Output Agreement Mechanisms and Common Knowledge}.
\newblock In {\em HCOMP 2014}.

\bibitem[\protect\citeauthoryear{Witkowski and
  Parkes}{2012a}]{witkowski2012robust}
Witkowski, J., and Parkes, D.
\newblock 2012a.
\newblock {A Robust Bayesian Truth Serum for Small Populations}.
\newblock In {\em AAAI 2012}.

\bibitem[\protect\citeauthoryear{Witkowski and
  Parkes}{2012b}]{witkowski2012peer}
Witkowski, J., and Parkes, D.
\newblock 2012b.
\newblock {Peer Prediction without a Common Prior}.
\newblock In {\em ACM EC 2012},  964--981.

\bibitem[\protect\citeauthoryear{Witkowski and Parkes}{2013}]{Witkowski_ec13}
Witkowski, J., and Parkes, D.~C.
\newblock 2013.
\newblock {Learning the Prior in Minimal Peer Prediction}.
\newblock In {\em the 3rd Workshop on SCUGC}.

\bibitem[\protect\citeauthoryear{Witkowski \bgroup et al\mbox.\egroup
  }{2013}]{Witkowski_hcomp13}
Witkowski, J.; Bachrach, Y.; Key, P.; and Parkes, D.~C.
\newblock 2013.
\newblock {Dwelling on the Negative: Incentivizing Effort in Peer Prediction}.
\newblock In {\em HCOMP 2013}.

\end{thebibliography}

\appendix

\onehalfspacing

\onecolumn
\section*{\LARGE{Appendix}}
\section{Randomized task assignment}

We explain on why we need a well structured random task assignment. First we make sure for each task, it has been assigned at least to two workers, so each of the assignment can serve as a peer evaluation for the other. Secondly for any pair of workers that share the same task, they also need to have distinct tasks, which is motivated by the mechanism (\texttt{DG13}). 

  A reader may notice that by simply assigning each task to \emph{all} worker both of above conditions will be satisfied satisfied. (For example, assign 4 tasks $\{1,2,3,4\}$ to all 4 workers, and when distinct tasks are needed for worker 1 and 2, we can compute using only task 3 and 4 for each worker respectively.) But instead we will make our assignment process random, which will help exclude the possibility of more complicated collusion strategies (e.g. such as colluding on subset of tasks but not on the rest, see the example below), especially when also considering we can randomly shuffle labels (or IDs) of both workers and tasks. Therefore  in this paper we only consider one type of collusion, that is if workers decide to contribute uninformative signals, they will report the same labels for all tasks. 

\paragraph{Example on more sophisticated collusion:} Suppose workers are assigned the same set of tasks with the same ID. And for simplicity we assume there are even number of tasks. Workers agree on the IDs of tasks, and they will agree on reporting -1 for odd number ID tasks, and +1 for even IDs. Then 
\begin{align*}
B_i(k) &= 1(L_i(k)=L_j(k)) - L^d_i \cdot L^d_j -\overline{L}^d_i \cdot \overline{L}^d_i\\
&=1-\frac{1}{2}\cdot\frac{1}{2}-\frac{1}{2}\cdot\frac{1}{2} =\frac{1}{2},
\end{align*}
which is the maximum score that can be achieved as 
\begin{align*}
 &1(L_i(k)=L_j(k))  \leq 1,\\
 &L^d_i \cdot L^d_j + \overline{L}^d_i \cdot \overline{L}^d_i \geq 1/2.
\end{align*}
To see this, denote by $x:=L^d_i, y:= L^d_j$ -- we know $0 \leq x,y \leq 1$. The following holds:
\begin{align*}
xy+(1-x)(1-y) \leq (x^2+y^2)/2 + ((1-x)^2+(1-y)^2)/2 = \frac{2x^2-2x+1}{2} + \frac{2y^2-2y+1}{2} \geq 1/2.
\end{align*}

\paragraph{A feasible assignment:}Now we demonstrate that such an assignment that meets all requirements can be achieved. For example, suppose we have $4$ workers, set $M=2, d=1, K=2$ and we assign 4 tasks $\{1,2,3,4\}$ each time as follows:
\begin{align*}
\text{Worker 1:}~\{1,2\},~\text{Worker 2:}~\{1,3\},~\text{Worker 3:}~\{2,4\},~\text{Worker 4:}~\{3,4\}~
\end{align*}
Not hard to verify the above assignment satisfies all constraints we enforced. More generally when we have $N$ workers, we can prepare $N$ tasks to assign for each time. Again set $M=2, d=1, K=2$. Denote the tasks received by worker $i$ as $t_i(1),t_i(2)$. Then the assignment can be adaptively decided as follows:
\begin{align*}
t_1(1) = &1,~t_1(2) = 2,~t_2(1) = 1, t_2(2) = 3,\\
t_i(1) =& t_{i-2}(2), t_i(2) = \min\{t_{i-1}(2)+1,N\}, \forall i > 2~.
\end{align*}
It is easy to verify the above assignment rule satisfies our requirements: each tasks is assigned at least twice; workers receiving the same task also receive different tasks; all tasks are assigned the same number of times; not all tasks are assigned to all workers. 


\section{Proof of Theorem \ref{bne:static}}
\begin{proof}
Denote by $\mathcal P_{+1}, \mathcal P_{-1}$ the priors for labels, and the probability of observing label +1 and -1 of each worker $i$ under effort level $e_i$ as 
$$
p_{+,e_i} := \Pr[L_i = +1|e_i] = \mathcal P_{+} p_{e_i} + \mathcal P_{-} (1-p_{e_i}),
$$
and $p_{-,e_i} := 1-p_{+,e_i}$. 
W.l.o.g. consider task $k$ of worker $i$ (when task $k$ is indeed assigned to worker $i$). Exerting effort or not on task $k$ will affect $u_i$ in two ways: 

First on $\mathbb E[B_i(k)]$: notice the decision on $k$ does not affect $\mathbb E[ L^d_i \cdot L^d_j]$. For $\mathbb E[ 1(L_i(k)=L_j(k))]$,
consider the fact that every other player is following the threshold policy that $e_j(k) = H$ if $c_j(k) \leq c^*$, and truthfully reporting. Then
\begin{align*}
\mathbb E[p_{e_j(k)}] &= F(c^*)p_H+(1-F(c^*))(1-p_L),~\\
\mathbb E[p_{+,e_j(k)}] &=F(c^*)p_{+,H}+(1-F(c^*))p_{+,L}.
\end{align*}
From which we have the utility difference between exerting and not exerting effort becomes: 
\begin{align*}
&\mathbb E[B_i(k)|e_i=H]  -\mathbb E[B_i(k)|e_i=L]\\
=& \mathbb E[p_{H} p_{e_j(k)} + (1-p_{H})(1-p_{e_j(k)})] \\
&~~~~- \mathbb E[p_{L} p_{e_j(k)} + (1-p_{L})(1-p_{e_j(k)})] \\
=&  (p_H-p_L) (2\mathbb E[p_{e_j(k)}]-1), ~
\end{align*}
where 
$$
\mathbb E[p_{e_j(k)}] = F(c^*)p_{H}+(1-F(c^*))p_{L}.
$$
Now consider the effect of $e_i(k)$ on $E[B_i(k')], k' \neq k$. Suppose after the randomized assignment, $k$ appears in $D \leq M-1$ other tasks' distinct set. Denote the set as $\mathcal D$. For $k' \in \mathcal D$, $e_i(k)$ affects the ``penalty term'': first we know by independence that 
\begin{align*}
\mathbb E[ L^d_i \cdot L^d_j] &= \mathbb E[ L^d_i]\cdot \mathbb  E[ L^d_j]~. 
\end{align*}
Then 
\begin{align*}
&\mathbb E[ L^d_i \cdot L^d_j|e_i(k)=H]- \mathbb E[ L^d_i \cdot L^d_j|e_i(k)=L] \\
 =& \mathbb E[ L^d_j] \biggl(\mathbb E[ L^d_i |e_i(k)=H]- \mathbb E[ L^d_i|e_i(k)=L]\biggr)\\
=&\mathbb E[p_{+,e_j}]\frac{p_{+,H}-p_{+,L}}{d}~.
\end{align*}
And similarly 
\begin{align*}
&\mathbb E[ \overline{L}^d_i \cdot \overline{L}^d_j|e_i(k)=H]- \mathbb E[ \overline{L}^d_i \cdot \overline{L}^d_j|e_i(k)=L] = (1-\mathbb E[p_{+,e_j}])\frac{p_{+,L}-p_{+,H}}{d}
\end{align*}
%
%
Summarize above difference we know:
\begin{align*}
&\mathbb E[u_i|e_i(k) = H]-\mathbb E[u_i|e_i(k)=L]=V_1 \cdot F(c^*)+V_2,
\end{align*}
where
\begin{align*}
&V_1:= 2(p_H-p_L)^2[1-\frac{D}{d}(\mathcal P_{+}-\mathcal P_{-})^2]\\
&V_2 := (2p_L-1) [1-\frac{D}{d}(\mathcal P_{+}-\mathcal P_{-})^2] (1-2\mathcal P_{-})^2~.
\end{align*}
The equilibrium equation establishes itself when the above equals to the cost $c^*$: (after re-arrange)
\begin{align}
B[V_1 \cdot F(c^*)+V_2] =c^*.\label{eqn:ne}
\end{align}
When $D$ is chosen such that $D \leq d$ we know $V_1, V_2>0$ (as $\mathcal P_{+}, \mathcal P_{-} > 0$).
We claim when $F(\cdot)$ is concave, there exists a unique solution if $p_L > 0.5$: first $\text{LHS}(c^*=0) > \text{RHS}(c^*=0)$, and when
\begin{align*}
&B[V_1 \cdot F(c^*=\bar{c})+V_2] \leq c_{\max},
\end{align*}
we have LHS of Eqn. (\ref{eqn:ne}) and the RHS intersects exactly once. So this unique intersecting point $c^* \leq c_{\max}$ is the unique solution to Eqn. (\ref{eqn:ne}), and o.w., we have $c^* \equiv c_{\max}$, that is $B$ is large enough so that exerting effort is always the best action to take. 


Also reporting by reverting the answer, i.e., $r_i = 0$, the probability of matching the true label becomes $1-p_H < p_L$, which leads to even less utility. So deviating from truthfully reporting is not profitable. 
\end{proof}

\section{Proof of Lemma \ref{collusion:price}}

\begin{proof}
This lemma is due to the fact that if everyone else is colluding to mis-lead the learner into believing a wrong price, a particular worker has no incentive to also do so: first his reported data will not affect his own price in the future. And as a rational worker, he should not care about the prices received by others. Due to the index rule we adopted, workers can do better than colluding: deviating from colluding to not exert effort possibly increases his current stage payment, when $p_L > 0.5$, and when cost $c_i$ is small enough:
\begin{align*}
&~~~~\mathbb E[B_i(k)|e_i=H] -\mathbb E[B_i(k)|e_i=L] \\
&=  (p_H-p_L) (2\mathbb E[p_{e_j(k)}]-1), \\
&=(p_H-p_L) (2p_{L}-1) > 0~.
\end{align*}

\end{proof}

\section{Lipschitz assumption on $\mathcal U(B)$}

Detailed argument for establishing the Lipschitz conditions can be similarly found in \cite{fullversion}. We briefly mention it here: when we use lower bound approximation for $\mathcal P(B)$ we have
\begin{align*}
\mathcal U(B) = 1-\text{exp}(-2(\bar{p}(B)-0.5)^2 M) - \eta \cdot B \cdot M \biggl[\mathcal P_{-}(2\bar{p}(B)-1)(\mathcal P_{+}(2\bar{p}(B)-1)+1) \biggr]. 
\end{align*}
The second part $\mathcal P_{-}(2\bar{p}(B)-1)(\mathcal P_{+}(2\bar{p}(B)-1)+1)$ is obtained by computing $ \mathbb E[B_i(k)]$. It is easy to see both $\text{exp}(-2(\bar{p}(B)-0.5)^2 N)$ and the second quadratic terms are Lipschitz in $\bar{p}(B)$. The rest remains to prove is that $\bar{p}(B)$ is also Lipschitz in $B$, as the composition of bounded Lipschitz functions are also Lipschitz. Since
\begin{align*}
\bar{p}(B):=F(c^*(B))p_H+[1-F(c^*(B))]p_L= (p_H-p_L)F(c^*(B)) + p_L,
\end{align*}
and as $F(\cdot)$ is concave and Lipschitz in $c$, we only need to prove that $c^*(B)$ is Lipschitz in $B$. The proof can be similarly established as Lemma 13.2 in \cite{fullversion} with similar assumptions, based on the equilibrium equation we characterized in the proof of Theorem \ref{bne:static}.

\section{Proof of Theorem \ref{thm:learn}}
\begin{proof}
First notice by triangle inequality we know
\begin{align*}
R(T) &= \sum_{t=1}^T |  \sum_{g \in \{1,-1\}} \omega_g \cdot \mathbb E[\mathcal U(B^{*}_{g,t})] - \mathcal U(B^*)|\\
&\leq \sum_{g \in \{1,-1\}} \omega_g \sum_{t=1}^T |\mathbb E[\mathcal U(B^{*}_{g,t})]-\mathcal U(B^*)|\\
&:=\sum_{g \in \{1,-1\}}\omega_g  R_g(T),
\end{align*}
i.e., the total regret is upper bounded by the weighted sum of each crowd's regret. Since the two learning processes for the two groups $G_1,G_{-1}$ parallel each other, we omit the $g$ super- or sub-script: we analyze the learning performance for each sub-group and the ones for the other one follows exactly the same. 

We separate the regret analysis into exploration regret and exploitation regret, that is the regret incurred at the exploration and exploitation stages respectively. We start with analyzing the exploitation regret. In order to characterize the exploitation regret, we need to characterize the estimation error for estimating each $U(B_k)$. First by mean value theorem we know $\exists p \in [\min\{p_m(B_k),\tilde{p}_{m,t}(B_k)\}, \max\{p_m(B_k),\tilde{p}_{m,t}(B_k)\}]$
\begin{align*}
|\bar{p}(B_k)-\tilde{p}_t(B_k)| &=\frac{|\sqrt{2p_m(B_k)-1}-\sqrt{2\tilde{p}_{m,t}(B_k)}|}{2} \\
&=\frac{1}{2\sqrt{1-2(1-p)}}|p_m(B_k)-\tilde{p}_{m,t}(B_k)|, \\
&\leq \frac{1}{2\sqrt{2p_L-1}}|p_m(B_k)-\tilde{p}_{m,t}(B_k)|.
\end{align*}
The inequality comes from the fact 
$p \geq \min\{p_m(B_k),\tilde{p}_{m,t}(B_k)\} \geq p_L > 0.5,
$ so the bound is well defined. \footnote{In the learning section we simply consider the case $p_L>0.5$. }  Now consider the estimation error in $p_m(B_k)$. At time $t$, an exploitation phase, there are $D(t)$ number of exploration/samples guaranteed, so the estimation for $p_m(B_k)$ at time $t$ satisfies:
\begin{align*}
\Pr&\biggl [|p_m(B_k)-\tilde{p}_{m,t}(B_k)| \geq \frac{1}{t^{\theta/2}} \biggr ] \leq 2\text{exp}(-2(\frac{1}{t^{\theta/2}})^2 t^{\theta}\log t) = \frac{2}{t^2}~.
\end{align*}
Then w.h.p., we have established that 
$$
|\bar{p}(B_k)-\tilde{p}_t(B_k)|  \leq \frac{ t^{-\theta/2}}{2\sqrt{2p_L-1}}.
$$
By Lipschitz condition we know
\begin{align*}
|\tilde{\mathcal U}(B_k)-\mathcal U(B_k)| \leq \frac{L_1 \cdot t^{-\theta/2}}{2\sqrt{2p_L-1}} ~.
\end{align*}
For any $B$ that falls in the same interval as $B_k$ we know:
\begin{align*}
|\mathcal U(B)&-\mathcal U(B_k)| \leq L_2  |B-B_k|\leq L_2 \bar{B}T^{-z}.
\end{align*}
Combine above arguments, we know for any $B$, we have the estimated utility function $\tilde{\mathcal U}(B_k)$ (the same interval as $B$) satisfy the follows
\begin{align*}
|\tilde{\mathcal U}(B_k)-\mathcal U(B)| \leq \frac{L_1 \cdot t^{-\theta/2}}{2\sqrt{2p_L-1}}+L_2 \bar{B}T^{-z}.
\end{align*}
Denote by
$B^*_t = \text{argmax}_{B_k, k=1,2,...,\lceil T^z \rceil} \tilde{\mathcal U}_t(B_k)~,
$
i.e., $B^*_t$ is the estimated optimal bonus level at time $t$ -- at any time $t$, by searching through all arm and we can find the one maximizes the utility function from the empirically estimated function the learner currently has. Then
\begin{align*}
&~~~~~\mathcal U(B^{*}_t) - \mathcal U(B^*) \\
&\geq \tilde{\mathcal U}(B^{*}_t)- \frac{L_1 \cdot t^{-\theta/2}}{2\sqrt{2p_L-1}} - \mathcal U(B^*)\\
&\geq \tilde{\mathcal U}(B^*_k) - \mathcal U(B^*)-   \frac{L_1 \cdot t^{-\theta/2}}{2\sqrt{2p_L-1}}\\
&\geq - \frac{L_1 \cdot t^{-\theta/2}}{2\sqrt{2p_L-1}}-L_2 \bar{B}T^{-z}~,
\end{align*}
where $B^*_k$ is the bonus level that is in the same interval as $B^*$. 
Combine with the fact $\mathcal U(\tilde{B}^*) \leq \mathcal U(B)$ we have
\begin{align*}
\mathcal U(B^*)-\mathcal U(B^*_t) \leq \frac{L_1 \cdot t^{-\theta/2}}{2\sqrt{2p_L-1}}+L_2 \bar{B}T^{-z}.
\end{align*}

Then the regret for exploration phases can be bounded as 
\begin{align*}
&\sum_{t=1}^T ( \frac{L_1 \cdot t^{-\theta/2}}{2\sqrt{2p_L-1}}+L_2 \bar{B}T^{-z}) + O(\sum_{t=1}^T \frac{2}{t^2})\leq \frac{L_1 \cdot C(\theta)}{2\sqrt{2p_L-1}}T^{1-\theta/2} + L_2 \cdot C(z) \bar{B} T^{1-z} + \text{const.} 
\end{align*}
where we have used the fact that for any $0 < \alpha < 1$, there exists a constant $C(\alpha)$ such that
$\sum_{t=1}^T \frac{1}{t^{\alpha}} \leq C(\alpha) T^{1-\alpha}~.
$
The total number of explorations are ($T^{\theta}\log T$ number of explorations needed for each of the $\lceil T^z \rceil$ arms)
\begin{align*}
\lceil T^z \rceil \cdot T^{\theta}\log T = \lceil T^{\theta+z}\log T \rceil~.
\end{align*}

Then the total accumulated regret incurred (exploration regret + exploitation regret) is bounded by
\begin{align*}
R(T)&  \leq  | \max \mathcal U - \min \mathcal U| \cdot\lceil T^{\theta+z}\log T \rceil  + \frac{L_1C(\theta)}{\sqrt{2p_L-1}}T^{1-\theta/2} + L_2  C(z)\bar{B} T^{1-z} + \text{const.} 
\end{align*}
The above regret achieves the optimal order when the three exponent term matches each other:
\begin{align*}
\theta + z = 1 - \theta/2, ~1-\theta/2 = 1 - z,
\end{align*}
which leads to the solution of $z=1/4, \theta=1/2$, which further results to the optimal order of regret $T^{3/4}$. 

\end{proof}

\subsection{The assumption that $p_L > 0.5$}

The assumption that $p_L>0.5$ is somewhat bothering. There are two places this assumption is needed. The first place is in proving the collusion proof for our learning mechanism (Lemma \ref{collusion:price}). The reason for this particular assumption is that when $p_L = 0.5$, the matching probability becomes independent of worker $i$'s labeling accuracy (as intuitively one's answer is compared to a random guess. So there is no reason to exert effort in such a case, regardless of workers' realized costs.)
This requirement can be relaxed by considering the fact a $0<\beta<1$ fraction of workers is independent of the collusion. 


The second place that such an assumption is needed is in bounding the estimation errors (when use the mean value theorem, we need a bounded graident). We hope to find other bounds to get around of the requirement $p_L > 0.5$.

\section{Performance for (\texttt{SPP\_PostPrice}) under non-game theoretical model}
\begin{theorem}
When run \emph{(\texttt{SPP\_PostPrice})} for the non-game theoretical model, $R(T)$ is bounded as:
\begin{align*}
R(T)& \leq   O(\lceil T^{\theta+z}\log T \rceil) +2C(\theta)L_1 (\frac{2}{(2p_L-1)^2}+\frac{1}{2p_L-1}) T^{1-\theta/2} + L_2 T^{1-z}+\text{const.}
\end{align*}
\label{thm:nongame}
\end{theorem}
\begin{proof}
Similar with the proof for Theorem \ref{thm:learn}, we are going to bound a number of estimation errors in evaluating $U(\cdot)$. First of all notice at time $t$, 
\begin{align*}
&|\frac{\sum_{i \in G_{g(i)}} p^{i}_{m}(B_k)}{|G_{g(i)}|}- \frac{\sum_{i \in G_{g(i)}} \tilde{p}^{i}_{m}(B_k)}{|G_{g(i)}|}|\leq \sum_{i \in G_{g(i)}}| p^{i}_{m}(B_k)- \tilde{p}^{i}_{m}(B_k)|/|G_{g(i)}|.
\end{align*}
Again when the algorithm enters exploitation phases (so $D(t)$ number of samples guaranteed for all arms) we have $\forall k$ by Chernoff bound that
\begin{align*}
\Pr\biggl [| p^{i}_{m}(B_k)- \tilde{p}^{i}_{m}(B_k)| \geq t^{\theta/2}\biggr] \leq 2/t^2,
\end{align*}
Then via union bound we know with probability being at least $1-\frac{N+1}{t^2}$($\forall i$):
\begin{align*}
&|\frac{\sum_{i \in G_{g(i)}} p^{i}_{m}(B_k)}{|G_{g(i)}|}- \frac{\sum_{i \in G_{g(i)}} \tilde{p}^{i}_{m}(B_k)}{|G_{g(i)}|}|\leq t^{-\theta/2}.
\end{align*}
Then we know the following holds:
\begin{align*}
&~~~~~|\tilde{p}_{i,t}(B_k)-p_i(B_k)|\\
& \leq |\frac{p^{i}_{m}(B_k)+\bar{p}(B_k)-1}{2\tilde{p}_t(B_k)-1}-p_i(B_k)|+|\frac{p^{i}_{m,t}(B_k)-p^{i}_m(B_k)+\tilde{p}_t(B_k)-\bar{p}(B_k)}{2\tilde{p}_t(B)-1}|\\
&\leq  |\frac{p^{i}_{m}(B_k)+\bar{p}(B_k)-1}{2\tilde{p}(B_k)-1}-p^i(B_k)| + \frac{2t^{-\theta/2}}{2p_L-1}\\
&\leq 2 \frac{p^{i}_{m}(B_k)+\bar{p}(B_k)-1}{(2p_L-1)^2}|\tilde{p}_{t}(B_k)-\bar{p}(B_k)|+ \frac{2t^{-\theta/2}}{2p_L-1} \\
&\leq \frac{2}{(2p_L-1)^2} \frac{L_1 \cdot t^{-\theta/2}}{2\sqrt{2p_L-1}}+\frac{2}{2p_L-1} t^{-\theta/2}\\
&=\frac{ t^{-\theta/2}}{2p_L-1}(\frac{L_1}{(2p_L-1)^{1.5}}+2)~.
\end{align*}
The second inequality uses the  concentration bound, as well as the fact that
$
\bar{p}(B) \geq p_L > \frac{1}{2}, \forall B~
$
so that the estimation should never go smaller than this quantity. The third inequality applies mean value theorem to function $1/(2x-1)$; while the last one uses fact $|p^{i}_{m}(B_k)+\bar{p}(B_k)-1| < 1$ and concentration bound. Then for any $B$, we have the estimated utility function $\tilde{\mathcal U}(B)$ satisfy the follows
\begin{align*}
|\tilde{\mathcal U}(B)-\mathcal U(B)| \leq &L_1 \frac{ t^{-\theta/2}}{2p_L-1}(\frac{L_1}{(2p_L-1)^{1.5}}+2)~.
\end{align*}
The rest of analysis is very similar to the one we presented for Theorem \ref{thm:learn}. We will not re-state the derivation details, but we are led to:
\begin{align*}
&\mathcal U(B^*_t) - \mathcal U(B^*)\geq -L_1  \frac{ t^{-\theta/2}}{2p_L-1}(\frac{L_1}{(2p_L-1)^{1.5}}+2)- L_2 \bar{B} \cdot T^{-z}~.
\end{align*}
And the exploitation regret accumulated up to time $t$ is bounded by
\begin{align*}
&\sum_t \biggl( L_1  \frac{ t^{-\theta/2}}{2p_L-1}(\frac{L_1}{(2p_L-1)^{1.5}}+2)+ L_2 \bar{B} \cdot T^{-z}\biggr)+\text{const.}\leq  \frac{C(\theta)L_1}{2p_L-1}(\frac{L_1}{(2p_L-1)^{1.5}}+2)T^{1-\theta/2}+ L_2C(z) T^{1-z}+\text{const.}
\end{align*}
And the exploration regret is upper bounded as
$O(T^z D(t)) = O(T^{z+\theta}\log T)$.
\end{proof}
\end{document}